\documentclass[a4paper]{amsart}

\makeatletter
\let\MakeUppercase\relax
\makeatother

\usepackage{amssymb,latexsym}
\usepackage{amsmath,amsthm}
\usepackage{graphicx}
\usepackage[usenames,dvipsnames]{color}
\usepackage{float}
\usepackage{tikz}
\usepackage[backref=page,colorlinks=true,linkcolor=Orchid,citecolor=Orchid,urlcolor=Orchid]{hyperref}

\newtheorem{theorem}{Theorem}
\newtheorem{lemma}{Lemma}

\newtheorem{corollary}{Corollary}
\newtheorem*{perfect fluid}{Vacuum and perfect fluid cases}

\DeclareMathOperator{\tr}{tr}

\begin{document}


\title[LRS Bianchi Type~VIII cosmologies with anisotropic matter]{Dynamics of locally rotationally symmetric \\ Bianchi type~VIII cosmologies \\ with anisotropic matter}

\author{\textsc{Gernot HEI}\ss\textsc{EL}}

\address{\href{http://gravity.univie.ac.at}{Gravitational Physics Group\\
	Faculty of Physics\\
	University of Vienna\\
	Austria}}
	
\email{\href{mailto:Gernot.Heissel@Mac.com}{Gernot.Heissel@Mac.com}}
	
\date{\today}

\begin{abstract}
This paper is a study of the effects of anisotropic matter sources on the qualitative evolution of spatially homogenous cosmologies of Bianchi type~VIII. The analysis is based on a dynamical system approach and makes use of an anisotropic matter family developed by \mbox{Calogero} and \mbox{Heinzle} which generalises perfect fluids and provides a measure of deviation from isotropy. Thereby the role of perfect fluid solutions is put into a broader context.

The results of this paper concern the past and future asymptotic dynamics of locally rotationally symmetric solutions of type~VIII with anisotropic matter. It is shown that solutions whose matter source is sufficiently close to being isotropic exhibit the same qualitative dynamics as perfect fluid solutions. However a high degree of anisotropy of the matter model can cause dynamics to differ significantly from the vacuum and perfect fluid case.
\end{abstract}

\maketitle



\section{Introduction and motivation}

For spatially homogenous (SH) spacetimes the Einstein-matter equations for a large variety of matter sources reduce to an autonomous system of ordinary differential equations in time. Thus the mathematical theory of dynamical systems can be applied to gain insights into the qualitative behaviour of SH solutions. This approach has been used in mathematical cosmology, e.g. to address questions relevant for observational cosmology, in particular concerning the role the Friedmann solutions play in the more general context of SH cosmologies that are not spatially isotropic in general. On the other hand the interest in SH models is nourished by the believe that the dynamics of SH cosmologies towards the initial singularity is crucial for the understanding of the behaviour of more general spacetimes close to singularities; cf.~\cite{Heinzle Uggla Roehr (2009)} and references therein.

Less is known about SH solutions with matter sources more general than perfect fluids. \mbox{Calogero} and \mbox{Heinzle} have developed a matter family naturally generalising perfect fluids that contains large classes of anisotropic matter sources and is suited for a dynamical system analysis. It includes a measure of the deviation from isotropy and thus allows to investigate the role of perfect fluid solutions in the more general context of solutions with anisotropic matter. The SH cosmologies considered were of Bianchi type~I, and locally rotationally symmetric (LRS) types~I, II and~IX; cf.~\cite{M&S Bianchi I,M&S}. While dynamical system analyses with specific anisotropic matter sources have been carried out before, the approach by \mbox{Calogero} and \mbox{Heinzle} can be regarded as a first step towards a systematic study of the effects of anisotropic matter to the qualitative dynamics of SH cosmologies.

This paper is concerned with the analysis of SH cosmologies of LRS Bianchi type~VIII with anisotropic matter. In section~\ref{S:Setup} the basic features of the anisotropic matter family are stated, and the state space and evolution equations for LRS type~VIII are given. The dynamical system analysis is performed in section~\ref{S:ds analysis}, where some technical details on the analysis of the flow at infinity are contained in the appendix. The results are given and discussed in section~\ref{S:Results}, where the main result is formulated in theorem~\ref{T:limit sets} and corollaries~\ref{C:past asymptotics} and~\ref{C:future asymptotics}. As a small byproduct, the results cover the future asymptotics for perfect fluids with $p=w\rho$ and $w\in\left(-\frac{1}{3},0\right)$, which might fill a little gap in the literature. Section~\ref{S:Vlasov} is concerned with an extension of the formalism to treat Vlasov matter dynamics with massive particles.

Part of the material needed in sections~\ref{S:Setup}, \ref{S:ds analysis} and~\ref{S:Vlasov} has already been presented in~\cite{M&S} or~\cite{M&S;Vlasov} to analyse LRS Bianchi types~I, II and~IX. At these points in the text, only the crucial steps and results are quoted from there.

\section{The LRS Bianchi type~VIII setup}\label{S:Setup}

In a frame $(\mathrm dt,\hat\omega^1,\hat\omega^2,\hat\omega^3)$ adapted to the symmetries, an LRS Bianchi class~A metric has the form
\begin{equation}\notag
^4g=-\,\mathrm dt\otimes\mathrm dt+g_{11}(t)\,\hat\omega^1\otimes\hat\omega^1+g_{22}(t)\,(\hat\omega^2\otimes\hat\omega^2+\hat\omega^3\otimes\hat\omega^3).
\end{equation} In type~VIII, $\mathrm d\hat\omega^i=-\frac{1}{2}\epsilon_{jkl}\hat n^{ij}\hat\omega^k\wedge\hat\omega^l$, with $[\hat n^{ij}]\equiv\mathrm{diag}(-1,1,1)$.
Greek indices denote spacetime components while Latin indices label spatial components w.r.t.\ the adapted frame. 
The metric will be subject to the Einstein equations---without cosmological constant---in geometrised units ($8\pi G=c=1$), cf.~\cite[Eq 2]{M&S}.

\subsection{The anisotropic matter family}\label{SS:matter family}

For a perfect fluid that is non-tilted with respect to $\mathrm dt$, the components of the stress-energy tensor are $[{T^\mu}_\nu]=\mathrm{diag}(\rho,p,p,p)$, where $\rho$ and $p$ denote the energy density and pressure of the fluid. It is an isotropic matter model since the eigenvalues of $[{T^i}_j]$ are all equal. When $p=w\rho$ with $w=\mathrm{const}$, the fluid is said to obey a linear equation of state.\footnote{For example, $w=0$ corresponds to dust and $w=\frac{1}{3}$ to radiation.}

The matter models considered in this paper form a family of models generalising perfect fluids with linear equation of state. This family of models is described in detail in~\cite[section~3]{M&S}. In the following, a brief description tailored to the present purposes is given:

The components of the stress-energy tensor are $[{T^\mu}_\nu]=\mathrm{diag}(\rho,p_1,p_2,p_2)$ where the energy density $\rho\geq0$ and the isotropic pressure $p:=\sum_i p_i/3$ obey a linear equation of state $p=w\rho$, with $w=\mathrm{const}$. Defining the dimensionless rescaled principal pressures as $w_i:=p_i/\rho$ it is clear that $w=\sum_i w_i/3$, which in turn implies that $w_1$ and $w_2$ are not independent once $w$ is given. As a consequence of the Einstein equations and some basic assumptions on the matter family given in~\cite[section~3]{M&S}, $w_2$ (and thus $w_1$) is a function of the quantity $s:=\frac{g^{22}}{\sum_i g^{ii}}\in\left(0,\frac{1}{2}\right)$. Clearly, $s$ gives a measure of anisotropy of the spatial metric components while $w_2$ encodes the anisotropy of the matter content.\footnote{In~\cite{M&S} $w_2(s)$ is called anisotropy function and denoted by $u(s)$.} The simple example of isotropic matter corresponds to $w_1=w_2=w$. Also, note that $s\to0$ and $s\to\frac{1}{2}$ correspond to $g_{22}\to\infty$ (or $g_{11}\to0$) and $g_{11}\to\infty$ (or $g_{22}\to0$) respectively, i.e. to a singular metric. A basic assumption is that the limit of $w_1$ for $s\to\frac{1}{2}$ coincides with the limit of $w_2$ for $s\to0$, i.e. that $w_1\left(\frac{1}{2}\right)=w_2(0)$. Hence one can define an anisotropy parameter
\begin{equation}\label{E:beta}
\beta:=2\frac{w-w_2(0)}{1-w}
\end{equation} that provides a measure of deviation from an isotropic matter state in the extremal cases where the metric is singular. Note that $\beta=0$ corresponds to matter models that behave like a perfect fluid close to singularities.

The analysis in section~\ref{S:ds analysis} will show that the qualitative dynamics of LRS Bianchi type~VIII solutions does not depend on the whole function $w_2$ but merely on the value $w_2(0)$ where the metric is singular. Thus any two matter models of the anisotropic matter family that share the same parameters $w$ and $\beta$ also share the same qualitative dynamics, even if they have different functions $w_i(s)$. Accordingly, $\beta$ serves as the parameter to investigate the influence of the anisotropy of the matter on the dynamics, even though it gives a precise measure of the anisotropy only close to singularities.
Therefore, in the context of the qualitative analysis of the dynamics of solution that follows, a specific class of matter models is simply characterised by a pair $(w,\beta)$ in the parameter space
\begin{equation}\notag
\mathbb P:=\left({\textstyle-\frac{1}{3}},1\right)\times\mathbb R.
\end{equation} Here $w$ is restricted to $\left(-\frac{1}{3},1\right)$ since the primary interest is in matter models that obey the standard energy conditions~\cite[p~218--220]{Wald}: The weak energy condition corresponds to $w_i\geq-1$. The strong energy condition requires the weak energy condition to hold and $w\geq-\frac{1}{3}$. The dominant energy condition is $|w_i|\leq1$. Therefore, by~\eqref{E:beta}, for the energy conditions to hold, one has to restrict $\beta$ to $\bigl[\mathrm{max}\bigl(-2,-\frac{1+w}{1-w}\bigr),1\bigr]$; cf.~\cite[table~2 in section~3.4]{M&S}. The dominant energy condition is only marginally satisfied when $\beta(w)$ takes the boundary values. The parameter space $\mathbb P$ and the subset for which the energy conditions are satisfied is illustrated in figure~\ref{F:parameter space} together with the bifurcation lines which will be explained in section~\ref{S:ds analysis}.

\subsection{The Einstein equations as a dynamical system}\label{SS:dynamical system}
LRS Bianchi types~VIII and~IX share the same evolution equations which have been derived in detail in~\cite{M&S}. This subsection gives a brief outline:

Due to spatial homogeneity, the Einstein equations for LRS type~VIII with the above matter source are a constrained autonomous system of ordinary differential equations in $t$ for the components of the spatial metric $g_{ij}$ and the extrinsic curvature $k_{ij}$. This system is then written in terms of quantities that are standard cosmological parameters and/or bring the equations into suitable shape. These quantities are the Hubble~scalar $H:=-\frac{1}{3}({k^1}_1+2{k^2}_2)$, the shear variable $\sigma_+:=\frac{1}{3}({k^1}_1-{k^2}_2)$ and the quantity $m_1:=\sqrt{g_{11}}/g_{22}$. Finally these are divided by the dominant variable $D:=\sqrt{H^2-1/(3g_{22})}>0$ to obtain the normalised variables $(H_D,\Sigma_+,M_1):=(H,\sigma_+,m_1)/D$.\footnote{In~\cite{M&S} this normalisation has been taken in lieu of the Hubble normalisation~\cite[section~5.2]{WE} mainly because it yields a compact LRS type~IX state space. Although this is not the case for LRS type~VIII, the dominant normalisation still has favourable properties. For example the type~III form of flat spacetime is a fixed point solution in this formulation; cf. table~\ref{Table:exact solutions}.}

The resulting representation of the Einstein equations for anisotropic matter filled LRS Bianchi type~VIII spacetimes is then given by the dynamical system
\begin{equation}\label{E:dsX}
\begin{bmatrix} H_D \\Ê\Sigma_+ \\ M_1 \end{bmatrix}' =
\begin{bmatrix}
	-(1-H_D^2) (q-H_D \Sigma_+) \\
	-(2-q)H_D\Sigma_+-(1-H_D^2)(1-\Sigma_+^2)+\frac{M_1^2}{3}+3\Omega\,(w_2(s)-w) \\
	M_1\left(qH_D-4\Sigma_++(1-H_D^2)\Sigma_+\right)
\end{bmatrix}
\end{equation}
which represents the evolution equations, and the constraints
\begin{equation}\label{E:constraints}
\Omega+\Sigma_+^2+\frac{M_1^2}{12}=1 \quad \text{and} \quad 1-H_D^2=-\frac{1}{3D^2g_{22}}<0.
\end{equation}\label{E:hc}The former is the Hamiltonian constraint, while the latter follows directly from the definition of $D$. $\Omega:=\frac{\rho}{3D^2}>0$ denotes the normalised energy density and $q:=2\Sigma_+^2+\frac{1+3w}{2}\Omega$ the deceleration parameter. The prime denotes derivatives with respect to rescaled time; $(\cdot)':=\frac{1}{D}\frac{\partial}{\partial t}(\cdot)$. Note that the definition of $M_1$ and the second constraint imply that $s$ can be regarded as a function of $H_D$ and $M_1$, $s=\left(2-3\frac{1-H_D^2}{M_1^2}\right)^{-1}$. Therefore~\eqref{E:dsX} is a closed system once $w_2(s)$ is prescribed through the `equation of state' of the anisotropic matter.

\subsection{The state space $\mathcal X_\mathrm{VIII}$}\label{SS:state space}
The LRS type~VIII state space is determined by the constraints~\eqref{E:constraints}: Since $\Omega>0$ and $M_1>0$ by definition, it follows from the Hamiltonian constraint that $\Sigma_+^2<1$ and $\Sigma_+^2+\frac{M_1^2}{12}<1$. The inequality $1-H_D^2<0$ imposed by the second constraint implies that the state space is the union of two disjoint sets: $H_D>1$ corresponds to positive $H$ and hence to forever expanding universes, while $H_D<-1$ corresponds to forever contracting universes. However, since~\eqref{E:dsX} is invariant under the reflection
$(t, H_D,\Sigma_+)\rightarrow-(t, H_D,\Sigma_+)$ it suffices to restrict to the expanding case. Accordingly the LRS Bianchi type~VIII state space is defined as
\begin{equation}\notag
\mathcal X_\mathrm{VIII}:=\left\{\begin{bmatrix} H_D \\Ê\Sigma_+ \\ M_1 \end{bmatrix}\in\mathbb R^3\Bigg|H_D\in(1,\infty),\Sigma_+\in(-1,1),M_1\in(0,\sqrt{12(1-\Sigma_+^2)})\right\}\,,
\end{equation}which forms the tunnel-like structure depicted in figure~\ref{F:state space X}. The boundary subsets of $\mathcal X_\mathrm{VIII}$ are $\mathcal V_\mathrm{VIII}$, $\mathcal B_\mathrm{III}$ and $\mathcal S_\sharp$, which correspond to $\Omega=0$, $M_1=0$ and $H_D=1$, respectively. LRS type~VIII vacuum solutions thus lie in $\mathcal V_\mathrm{VIII}$ while solutions in $\mathcal S_\sharp$ and $\mathcal B_\mathrm{III}$ are of Bianchi types~II and~III respectively; cf.~\cite[sections~9.2 and the first remark in~10.1]{M&S}. $\mathcal S_\infty$, which corresponds to $H_D\rightarrow\infty$, is of course not a boundary, but one can think of it as boundary in a compactified version of the state space; cf. appendix~\ref{S:appendix}.
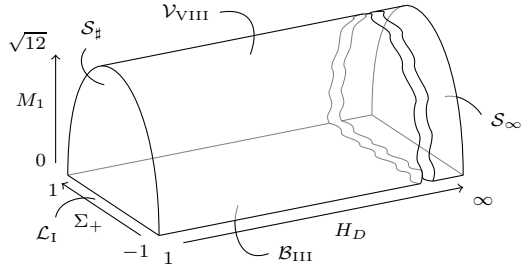
\begin{figure}\begin{tikzpicture}[scale=.8]

\node[left]at(1,3.7){$\scriptstyle{\mathcal V_\mathrm{VIII}}$};					
\draw(1,3.7)..controls+(10:.5)and+(70:.5)..(1.5,3);								%
\node[right]at(1.8,-.3){$\scriptstyle{\mathcal B_\mathrm{III}}$};				%
\draw(1.8,-.3)..controls+(200:.5)and+(-95:.2)..(1.25,.25);						%
\draw[help lines](1.25,.25)..controls+(85:.1)and+(250:.1)..(1.3,.5);			%
\node[right]at(5.3,1.9){$\scriptstyle\mathcal S_\infty$};							%
\draw(5.3,2)..controls+(120:.15)and+(-10:.15)..(4.885,2.3);					%
\draw[help lines](4.885,2.3)..controls+(170:.05)and+(20:.05)..(4.7,2.3);	%
\node[above]at(-1.1,3.1){$\scriptstyle\mathcal S_\sharp$};					%
\draw(-1.2,3.1)..controls+(-130:.3)and+(170:.3)..(-.9,2.5);						%
\node[below left]at(-1.5,.3){$\scriptstyle\mathcal L_\mathrm I$};				%
\draw(-1.7,.3)..controls+(80:.3)and+(160:.2)..(-1.05,.6);								%

\draw[help lines](5,1)--(3.5,2)--(3,1.9);					
\draw[help lines](2.8,1.86)--(-1.5,1);						%

\draw[help lines](4.05,3.812)..controls+(185:.16)and+(90:1.6)..(3.5,2);	
\draw(5,1)..controls+(90:1.85)and+(5:.395)..(4.05,3.812);						%
\draw(0,0)..controls+(90:3)and+(90:3)..(-1.5,1);									

\draw(4.5,.9)														
	..controls+(70:.15)and+(-60:.15)..(4.46,1.24)		%
	..controls+(120:.15)and+(-60:.15)..(4.47,1.74)		%
	..controls+(120:.15)and+(-110:.15)..(4.42,2.24)	%
	..controls+(70:.15)and+(-100:.15)..(4.2,2.84)		%
	..controls+(80:.2)and+(-100:.15)..(4,3.34)			%
	..controls+(80:.15)and+(-120:.15)..(3.7,3.69)		%
	..controls+(70:.05)and+(0:.05)..(3.55,3.71)			%
	..controls+(183:.05)and+(15:.03)..(3.445,3.68);	%
\draw[help lines](3.45,3.68)									%
	..controls+(183:.03)and+(170:.03)..(3.4,3.64)		%
	..controls+(-10:.05)and+(70:.1)..(3.2,3.44)			%
	..controls+(250:.1)and+(80:.1)..(3.15,3.14)			%
	..controls+(260:.1)and+(120:.1)..(3.05,2.79)		%
	..controls+(300:.1)and+(70:.1)..(3,2.44)				%
	..controls+(250:.1)and+(60:.2)..(3,2.09)				%
	..controls+(240:.1)and+(110:.1)..(3,1.9)				%
	..controls+(-15:.1)and+(85:.1)..(3.2,1.74)			%
	..controls+(265:.08)and+(80:.1)..(3.5,1.54)			%
	..controls+(260:.05)and+(75:.1)..(3.75,1.39)		%
	..controls+(255:.1)and+(85:.1)..(4.1,1.19)			%
	..controls+(265:.15)and+(80:.15)..(4.33,1.01);		%
\draw(4.33,1.01)													%
	..controls+(260:.1)and+(170:.1)..(4.5,.9);				%
\draw(4.3,.86)														
	..controls+(70:.15)and+(-60:.15)..(4.26,1.2)			%
	..controls+(120:.15)and+(-60:.15)..(4.27,1.7)		%
	..controls+(120:.15)and+(-110:.15)..(4.22,2.2)		%
	..controls+(70:.15)and+(-100:.15)..(4,2.8)			%
	..controls+(80:.2)and+(-100:.15)..(3.8,3.3)			%
	..controls+(80:.15)and+(-120:.15)..(3.5,3.65)		%
	..controls+(70:.05)and+(0:.05)..(3.35,3.67);			%
\draw[help lines](3.35,3.67)									%
	..controls+(180:.05)and+(170:.05)..(3.2,3.6)		%
	..controls+(-10:.05)and+(70:.1)..(3,3.4)				%
	..controls+(250:.1)and+(80:.1)..(2.95,3.1)			%
	..controls+(260:.1)and+(120:.1)..(2.85,2.75)		%
	..controls+(300:.1)and+(70:.1)..(2.8,2.4)				%
	..controls+(250:.1)and+(60:.2)..(2.8,2.05)			%
	..controls+(240:.1)and+(110:.1)..(2.8,1.86)			%
	..controls+(-15:.1)and+(85:.1)..(3,1.7)					%
	..controls+(265:.08)and+(80:.1)..(3.3,1.5)			%
	..controls+(260:.05)and+(75:.1)..(3.55,1.35)		%
	..controls+(255:.1)and+(85:.1)..(3.9,1.15)			%
	..controls+(265:.15)and+(80:.15)..(4.13,.97)		%
	..controls+(260:.1)and+(170:.1)..(4.3,.86);			%
	
\draw(4.5,.9)--(5,1);												
\draw(0,0)--(4.3,.86);											%
\draw(-1.5,1)--(0,0);								%

\draw(-.95,2.812)--(3.35,3.672);								
\draw(3.55,3.712)--(4.05,3.812);							%

\draw[->](.4,-.12)--(5,.8);										
\node[below left]at(.4,-.12){$\scriptstyle 1$};			%
\node[below right]at(5,.8){$\scriptstyle \infty$};			%
\node[below right]at(2.725,.35){$\scriptstyle H_D$};	%
\draw[->](-.3,0)--(-1.6,.867);									
\node[below left]at(0,0){$\scriptstyle -1$};				%
\node[below left]at(-1.5,1){$\scriptstyle 1$};				%
\node[below left]at(-.75,.5){$\scriptstyle\Sigma_+$};	%
\draw[->](-1.7,1.2)--(-1.7,3);									
\node[above left]at(-1.7,1){$\scriptstyle 0$};				%
\node[above left]at(-1.7,2.9){$\scriptstyle\sqrt{12}$};	%
\node[above left]at(-1.7,1.95){$\scriptstyle M_1$};	%

\end{tikzpicture}
\caption{The state space $\mathcal X_\mathrm{VIII}$ and its boundary subsets.}
\label{F:state space X}
\end{figure}

In this formulation LRS~type IX cosmologies are subject to the same evolution equations~\eqref{E:dsX}. However there is a sign change in the second constraint of~\eqref{E:constraints} leading to a different state space $\mathcal X_\mathrm{IX}$ for which $H_D\in(-1,1)$. Figuratively speaking $\mathcal X_\mathrm{IX}$ builds the tunnel connecting $\mathcal X_\mathrm{VIII}$ with the second disjoint LRS type~VIII set. Hence both $\mathcal X_\mathrm{VIII}$ and $\mathcal X_\mathrm{IX}$ share the boundary $\mathcal S_\sharp$. Cf.~\cite[section~9]{M&S}.

There are two challenges in connection with the analysis of the dynamical system~\eqref{E:dsX} in $\mathcal X_\mathrm{VIII}$. First, $\mathcal X_\mathrm{VIII}$ is not compact. Hence one has to perform a careful analysis of the `flow at infinity'. In the present case it will be shown that there do not exist orbits that emanate from or escape to infinity. Second, the dynamical system~\eqref{E:dsX} does not extend to the line $\mathcal L_\mathrm I$ in $\partial\mathcal X_\mathrm{VIII}$ for that $H_D=1$ and $M_1=0$ since $s(H_D,M_1)$ has no limit when $\mathcal L_\mathrm I$ is approached from $\mathcal X_\mathrm{VIII}$.\footnote{However $s(H_D,M_1)$ has limits when $\mathcal L_\mathrm I$ is approached from $\mathcal B_\mathrm{III}$ or $\mathcal S_\sharp$; cf. sections~\ref{SS:analysis on B} and~\ref{SS:analysis on S}.} The analogous problem occurs in the context of LRS type~IX, where it was overcome by introducing another set of coordinates that give regular access to this part of the boundary by `blowing up' $\mathcal L_\mathrm I$; cf.~\cite[section~10.2]{M&S}. This method can be adapted to the LRS type~VIII case by applying the same coordinate transformation. However the corresponding state space is again different:

\subsection{The state space $\mathcal Y_\mathrm{VIII}$}\label{SS:state space Y}

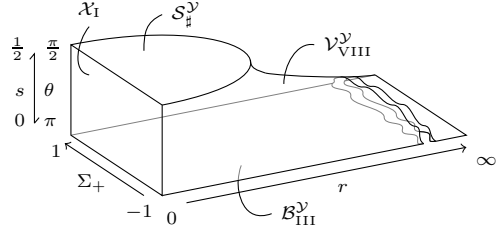
\begin{figure}\begin{tikzpicture}[scale=.8]

\node[right]at(1.8,-.3){$\scriptstyle{\mathcal B_\mathrm{III}^\mathcal Y}$};		
\draw(1.8,-.3)..controls+(200:.5)and+(-95:.2)..(1.25,.25);								%
\draw[help lines](1.25,.25)..controls+(85:.1)and+(250:.1)..(1.3,.5);					%
\node[right]at(0,3){$\scriptstyle\mathcal S_\sharp^\mathcal Y$};						%
\draw(0,3)..controls+(205:.3)and+(90:.3)..(-.3,2.4);										%
\node[above]at(-1.2,2.8){$\scriptstyle\mathcal X_\mathrm I$};						%
\draw(-1.2,2.8)..controls+(-130:.3)and+(170:.3)..(-1.2,2);								%
\node[right]at(2.5,2.5){$\scriptstyle{\mathcal V_\mathrm{VIII}^\mathcal Y}$};	%
\draw(2.5,2.5)..controls+(170:.3)and+(90:.3)..(2,1.85);									%

\draw[help lines](3.5,2)--(3,1.9);													
\draw[help lines](2.8,1.86)--(-1.5,1);											%


\draw[help lines](3.13,1.965)..controls+(190:.2)and+(25:.1)			
	..(3,1.9)..controls+(-120:.15)and+(80:.15)								%
	..(3.3,1.69)..controls+(260:.15)and+(90:.15)							%
	..(3.85,1.34)..controls+(-90:.15)and+(75:.15)							%
	..(4.15,1.14)..controls+(255:.15)and+(80:.15)							%
	..(4.4,.97)..controls+(269:.1)and+(180:.1)..(4.5,.9);					%
\draw(4.23,1.06)..controls+(-5:.09)and+(80:.09)..(4.4,.97)				%
	..controls+(269:.1)and+(180:.1)..(4.5,.9);									%
\draw(4.5,.9)..controls+(150:.15)and+(-5:.15)								%
	..(4.3,1.16)..controls+(175:.15)and+(0:.15)								%
	..(3.95,1.51)..controls+(180:.2)and+(5:.2)								%
	..(3.45,1.81)..controls+(185:.15)and+(10:.15)..(3.13,1.969);		%
\draw[help lines](2.93,1.965)..controls+(190:.2)and+(25:.1)			
	..(2.8,1.86)..controls+(-120:.15)and+(80:.15)							%
	..(3.1,1.65)..controls+(260:.15)and+(90:.15)							%
	..(3.65,1.3)..controls+(-90:.15)and+(75:.15)								%
	..(3.95,1.1)..controls+(255:.15)and+(80:.15)							%
	..(4.2,.93)..controls+(269:.1)and+(180:.1)..(4.3,.86);					%
\draw(4.3,.86)..controls+(150:.15)and+(-5:.15)								%
	..(4.15,1.15)..controls+(175:.15)and+(0:.15)							%
	..(3.8,1.5)..controls+(180:.2)and+(5:.2)									%
	..(3.3,1.8)..controls+(185:.15)and+(10:.15)..(2.93,1.965);			%

\draw(4.5,.9)--(5,1)--(3.5,2);														
\draw(0,0)--(4.3,.86);																%
\draw(-1.5,1)--(0,0)--(0,1.5)--(-1.5,2.5)--cycle;								
\draw(0,1.5)..controls+(11.31:2.8)and+(11.31:2.8)..(-1.5,2.5);			

\draw(1.7,2.03)..controls+(-10:.3)and+(190:.4)..(3.5,2);					
\draw(1.45,2.25)..controls+(-90:.15)and+(160:.15)..(1.7,2.03);		%

\draw[->](.4,-.12)--(5,.8);												
\node[below left]at(.4,-.12){$\scriptstyle0$};						%
\node[below right]at(5,.8){$\scriptstyle\infty$};					%
\node[below right]at(2.725,.35){$\scriptstyle r$};				%
\draw[->](-.3,0)--(-1.6,.867);											
\node[below left]at(0,0){$\scriptstyle -1$};						%
\node[below left]at(-1.5,1){$\scriptstyle 1$};						%
\node[below left]at(-.75,.5){$\scriptstyle\Sigma_+$};			%
\draw(-2.1,1.2)--(-2.1,2.35)--+(230:.12);							
\draw(-2.1,1.2)--+(50:.12);											%
\node[above left]at(-2.1,1){$\scriptstyle0$};						%
\node[above left]at(-2.1,2.1){$\scriptstyle\frac{1}{2}$};		%
\node[above left]at(-2.1,1.5){$\scriptstyle s$};					%
\node[above right]at(-2.1,1){$\scriptstyle\pi$};					%
\node[above right]at(-2.1,2.1){$\scriptstyle\frac{\pi}{2}$};	%
\node[above right]at(-2.1,1.5){$\scriptstyle\theta$};			%

\end{tikzpicture}
\caption{The state space $\mathcal Y_\mathrm{VIII}$ and its boundary subsets.}
\label{F:state space Y}
\end{figure}

The coordinate transformation used to analyse $\mathcal L_\mathrm I$ is
\begin{equation}\label{E:coord trafo}
1-H_D^2=2r\cos\theta,\quad M_1^2=3r\sin\theta,\quad\Sigma_+\text{ unchanged},
\end{equation}
where $r\ge0$ and $\theta\in[\frac{\pi}{2},\pi]$. Thereby, $\mathcal X_\mathrm{VIII}$ is transformed to the state space
\begin{equation}\notag
\mathcal Y_\mathrm{VIII}:=\left\{\begin{bmatrix} r \\\theta \\ \Sigma_+ \end{bmatrix}\in\mathbb R^3\Bigg|r\in\left(0,\frac{4(1-\Sigma_+^2)}{\sin\theta}\right),\theta\in\left(\frac{\pi}{2},\pi\right),\Sigma_+\in(-1,1)\right\}
\end{equation} which is depicted in figure~\ref{F:state space Y}. From~\eqref{E:coord trafo} one has the following correspondences between subsets of $\overline{\mathcal X}_\mathrm{VIII}$ and $\overline{\mathcal Y}_\mathrm{VIII}$:
\begin{equation}\label{E:correspondences}
\mathcal X_\mathrm{VIII}\sim\mathcal Y_\mathrm{VIII},\quad\overline{\mathcal V}_\mathrm{VIII}\sim\overline{\mathcal V}_\mathrm{VIII}^\mathcal Y,\quad\overline{\mathcal B}_\mathrm{III}\sim\overline{\mathcal B}_\mathrm{III}^\mathcal Y\cup\overline{\mathcal X}_I,\quad\overline{\mathcal S}_\sharp\sim\overline{\mathcal S}_\sharp^\mathcal Y\cup\overline{\mathcal X}_I,\quad\overline{\mathcal L}_\mathrm I\sim\overline{\mathcal X}_\mathrm I.
\end{equation}$\mathcal S_\infty$ corresponds to a line at infinity in the context of $\mathcal Y_\mathrm{VIII}$. Furthermore, \eqref{E:coord trafo} defines a diffeomorphism between $\overline{\mathcal X}_\mathrm{VIII}\backslash\overline{\mathcal L}_\mathrm I$ and $\overline{\mathcal Y}_\mathrm{VIII}\backslash\overline{\mathcal X}_\mathrm I$, from which it follows that the flows in these sets are topologically equivalent. In particular,
\begin{equation}\label{E:isomorphisms}
\mathcal X_\mathrm{VIII}\cong\mathcal Y_\mathrm{VIII},\quad\mathcal V_\mathrm{VIII}\cong\mathcal V_\mathrm{VIII}^\mathcal Y,\quad\mathcal B_\mathrm{III}\cong\mathcal B_\mathrm{III}^\mathcal Y\quad\text{and}\quad\mathcal S_\sharp\cong\mathcal S_\sharp^\mathcal Y.
\end{equation}
In contrast, the coordinate transformation performs a blowup of the line $\mathcal L_\mathrm I$ to the two-dimensional set $\mathcal X_\mathrm I$ defined by setting $r=0$. It can be identified with the LRS Bianchi type~I state space; cf. \cite[section~10.2]{M&S}. With~\eqref{E:coord trafo}, $s$ can be regarded as a function of $\theta$ alone. Therefore $s$ has a limit as $r\to0$, which implies that the dynamical system, when expressed in the coordinates $(r,\theta,\Sigma_+)$, extends regularly to $\mathcal X_\mathrm I$.

Finally, note that since $s(\theta)$ is a bijection on $[\frac{\pi}{2},\pi]$ one can as well choose $s$ as coordinate instead of $\theta$; hence the two labels on the axis in figure~\ref{F:state space Y}.


\section{The dynamical system analysis}\label{S:ds analysis}

Whenever possible the analysis is carried out in the coordinates $(H_D,\Sigma_+,M_1)$, which is in all sets except $\mathcal X_\mathrm I$. However the final results presented in section~\ref{S:Results} have to be interpreted in the context of the state space $\mathcal Y_\mathrm{VIII}$ since the system~\eqref{E:dsX} is not regular in $\overline{\mathcal X}_\mathrm{VIII}$, which prevents a complete global analysis in the original state space.

Since the matter parameters $w$ and $\beta$ enter the evolution equations everywhere except in the vacuum boundary and at infinity, cf. sections~\ref{SS:analysis on V} and~\ref{SS:Analysis on X}, their values determine the qualitative properties of the flow. For instance there are fixed points that only occur in the state space iff $(w,\beta)$ is in a certain subset of $\mathbb P$. Similarly, fixed points may have different local stability properties depending on $(w,\beta)$. The curves $\beta(w)\in\mathbb P$ dividing $\mathbb P$ into these subsets corresponding to qualitatively different dynamics shall be called bifurcation lines. These bifurcation lines will be found in the subsequent subsections and are plotted in figure~\ref{F:parameter space}.

\begin{figure}
\begin{tikzpicture}[xscale=5,yscale=3]
\path(7,0)coordinate(0;0);	\path(7,-.333)coordinate(0;-1);								
\path(6.666,-.666)coordinate(-1/3;-2);\path(8,-.666)coordinate(1;-2);				%
\path(6.666,-.333)coordinate(-1/3;-1);\path(8,-.333)coordinate(1;-1);				%
\path(6.666,0)coordinate(-1/3;0);\path(8,0)coordinate(1;0);							%
\path(6.666,.333)coordinate(-1/3;1);\path(8,.333)coordinate(1;1);					%
\path(6.666,.666)coordinate(-1/3;2);\path(8,.666)coordinate(1;2);					%
\path(6.775,-.089)coordinate(upeak);\path(6.666,-.167)coordinate(-1/3;-1/2);	%

\draw[draw=black!00,fill=black!8](-1/3;-1/2)								
	..controls+(-15:.1)and+(145:.1)..(0;-1)									%
	..controls+(-35:.1)and+(120:.1)..(7.333,-.666)						%
	--(1;-2)--(1;1)--(-1/3;1);														%
\draw[draw=black!00,fill=black!20](-1/3;-1/2)							
	..controls+(10:.05)and+(-90:.03)..(upeak)							%
	..controls+(90:.03)and+(-10:.05)..(-1/3;0)--(-1/3;-1/2);			%
\node[below left]at(6.5,-.1){$\scriptstyle\mathbb P_\supset$};		%
\draw(6.48,-.15)..controls+(30:.1)and+(140:.1)..(6.7,-.1);			%

\draw(-1/3;-1/2)																	
	..controls+(-15:.1)and+(145:.1)..(0;-1)									%
	..controls+(-35:.1)and+(120:.1)..(7.333,-.666);						%
\node[right]at(7.35,-.5){$\scriptstyle-\frac{1+w}{1-w}$};				%
\draw(7.35,-.5)..controls+(160:.05)and+(20:.05)..(7.24,-.51);		%

\draw[->](6.5,0)--(8.333,0);\node[below]at(8.333,0){$\scriptstyle w$};					
\draw[->](7,-1.2)--(7,1.2);\node[right]at(7,1.2){$\scriptstyle\beta$};						

\draw(-1/3;-2)--(1;-2);\node[left]at(-1/3;-2){$\scriptstyle-2$};	
\draw(-1/3;-1)--(1;-1);\node[left]at(-1/3;-1){$\scriptstyle-1$};	%
\draw(-1/3;1)--(1;1);\node[left]at(-1/3;1){$\scriptstyle1$};		%
\draw(-1/3;2)--(1;2);\node[left]at(-1/3;2){$\scriptstyle2$};		%
\node[left]at(-1/3;-1/2){$\scriptstyle-\frac{1}{2}$};					%

\draw(6.666,-1)--(6.666,1);\node[below]at(6.666,-1){$\scriptstyle-\frac{1}{3}$};		
\node[below]at(7.333,-1){$\scriptstyle\frac{1}{3}$};											%
\draw(8,-1)--(8,1);\node[below]at(8,-1){$\scriptstyle1$};										%

\draw(-1/3;0)..controls+(20:.4)and+(185:.4)..(1;1);				
\draw(-1/3;-1/2)..controls+(10:.05)and+(-90:.03)..(upeak)		%
	..controls+(90:.03)and+(-10:.05)..(-1/3;0);						%
\draw(-1/3;-1/2)..controls+(10:.1)and+(220:.2)..(0;0)				%
	..controls+(40:.4)and+(250:.4)..(7.666,1);						%
\node[above]at(7.77,1.02){$\scriptstyle\beta_\flat(w)$};			%
\draw(7.77,1.02)..controls+(-80:.06)and+(-20:.06)..(7.65,.9);	%
\node[right]at(7.64,.45){$\scriptstyle\beta_\sharp(w)$};			%
\draw(7.64,.45)..controls+(-170:.05)and+(90:.07)..(7.5,.28);	%
\node[right]at(7.07,-.167){$\scriptstyle\beta_\pm(w)$};			%
\draw(7.07,-.167)..controls+(180:.1)and+(10:.1)..(6.8,-.07);	%

\draw[dashed](6.666,-.044)..controls+(20:.4)and+(185:.4)..(1;1);	
\draw[dashed](6.666,.081)..controls+(-10:.1)and+(195:.1)..(6.9,.1)	%
	..controls+(15:.3)and+(248:.4)..(7.6,1);									%
\node[above]at(7.42,1.02){$\scriptstyle\beta_\flat^*(w)$};				%
\draw(7.42,1.02)..controls+(-90:.05)and+(170:.05)..(7.52,.9);			%
\node[right]at(7.64,.167){$\scriptstyle\beta_\sharp^*(w)$};				%
\draw(7.64,.167)..controls+(-120:.05)and+(-80:.05)..(7.5,.2);			%

\end{tikzpicture}
\caption{The bifurcation diagram in the parameter space $\mathbb P$. The shaded region (including $\mathbb P_\supset$) marks the subset for which the energy conditions are satisfied.}
\label{F:parameter space}
\end{figure}
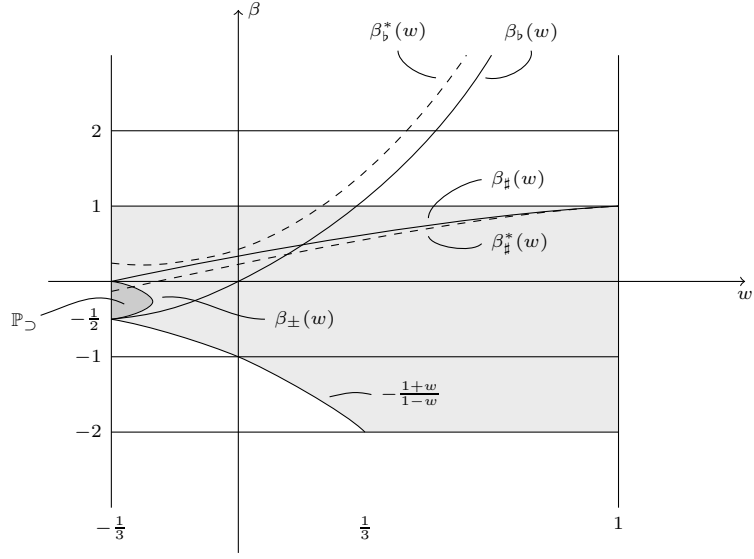

\subsection{Analysis in $\overline{\mathcal V}_\mathrm{VIII}$}\label{SS:analysis on V}

The dynamical system in $\overline{\mathcal V}_\mathrm{VIII}$ is obtained by setting $\Omega=0$ in~\eqref{E:dsX} and using~\eqref{E:constraints}:
\begin{equation}\label{E:dsV}
\begin{bmatrix} H_D \\ \Sigma_+ \end{bmatrix}' =
\begin{bmatrix} (1-H_D^2)(H_D-2\Sigma_+)\Sigma_+ \\ (1-\Sigma_+^2)\left(2+(1-\Sigma_+^2)+(H_D-\Sigma_+)^2\right) \end{bmatrix}
\end{equation}

There are three fixed points in $\overline{\mathcal V}_\mathrm{VIII}$, $T:=[1,-1]^\mathrm T$, $Q:=[1,1]^\mathrm T$ and $D:=[2,1]^\mathrm T$. The eigenvectors and eigenvalues of the linearisation of~\eqref{E:dsV} at the fixed points determine their local stability properties. One finds\footnote{Here and henceforth the notation follows the pattern $P:\begin{bmatrix}\lambda_1\\\lambda_2\end{bmatrix}\begin{bmatrix}v_1, v_2 \end{bmatrix}$ where $\lambda_i$ and $v_i$ denote the $i^\mathrm{th}$ eigenvalue and eigenvector of the linearisation of the dynamical system at $P$.}
\begin{equation}\notag
T: \begin{bmatrix} 6 \\ 12 \end{bmatrix} \begin{bmatrix} 1 & 0 \\ 0 & 1 \end{bmatrix}, \quad
Q: \begin{bmatrix} 2 \\ -4 \end{bmatrix} \begin{bmatrix} 1 & 0 \\ 0 & 1 \end{bmatrix} \quad \text{and} \quad
D: \begin{bmatrix} -3 \\ -6 \end{bmatrix} \begin{bmatrix} 1 & -2 \\ 0 & 1 \end{bmatrix},
\end{equation}
so $T$ is a source, $Q$ is a saddle repelling in $H_D$ direction and $D$ is a sink in $\overline{\mathcal V}_\mathrm{VIII}$.

It is proven in appendix~\ref{SS:appendix2dim} that there are exactly two more fixed points at infinity, $T_\infty=[\infty,-1]^\mathrm T$ and $Q_\infty=[\infty,1]^\mathrm T$. The stability properties then follow directly from~\eqref{E:dsV}: First, $\Sigma_+'>0$ in $\mathcal V_\mathrm{VIII}$. Second, for $H_D>2$ (and thus at infinity), $H_D' \gtreqqless 0$ for $\Sigma_+\lesseqqgtr0$. Hence $T_\infty$ and $Q_\infty$ play the role of saddles for the flow in $\overline{\mathcal V}_\mathrm{VIII}$ given in figure~\ref{F:flow on V}.

To interpret this as flow in $\overline{\mathcal V}_\mathrm{VIII}^\mathcal Y$, note that the fixed point $T$ in $\overline{\mathcal V}_\mathrm{VIII}$ corresponds to the closure of the line $T_\flat\to T_\sharp$ in $\overline{\mathcal V}_\mathrm{VIII}^\mathcal Y$. Since $T_\flat$ and $T_\sharp$ turn out to act as a source and a saddle in $\overline{\mathcal V}_\mathrm{VIII}^\mathcal Y$ in all cases respectively, cf. figures~\ref{F:flow13} to~\ref{F:flow1}, the orbits in $\mathcal V_\mathrm{VIII}^\mathcal Y$ have the form $T_\flat\to D$.

\begin{figure}\begin{tikzpicture}
\path (1,-1) coordinate (T);\node [below left] at (T) {$\scriptstyle T$};
\path (1,1) coordinate (Q);\node [above left] at (Q) {$\scriptstyle Q$};
\path (6,-1) coordinate (Tinfty);\node [below right] at (Tinfty) {$\scriptstyle{T_\infty}$};
\path (6,1) coordinate (Qinfty);\node [above right] at (Qinfty) {$\scriptstyle{Q_\infty}$};
\path (2,1) coordinate (S);\node [above] at (S) {$\scriptstyle D$};
\path (5.7,-1) coordinate (Tinterrupt1);\path (5.8,-1) coordinate (Tinterrupt2);
\path (5.7,1) coordinate (Qinterrupt1);\path (5.8,1) coordinate (Qinterrupt2);
\draw [->] (T)--(3.5,-1);\draw(3.5,-1)--(Tinterrupt1);\draw(Tinterrupt2)--(Tinfty); 
\draw [->] (Tinfty)--(6,0);\draw (6,0)--(Qinfty); 												
\draw (Qinfty)--(Qinterrupt2);\draw [->] (Qinterrupt1)--(4,1);\draw(4,1)--(S); 		
\draw [->] (T)--(1,0);\draw (1,0)--(Q);\draw [->] (Q)--(1.5,1);\draw (1.5,1)--(S); 	

\draw[->](.5,-.8)--(.5,.8);\node[left]at(.5,-.8){$\scriptstyle{-1}$};\node[left]at(.5,.8){$\scriptstyle1$};\node[left]at(.5,0){$\scriptstyle{\Sigma_+}$};
\draw[->](1.2,-1.5)--(5.8,-1.5);\node[below]at(1.2,-1.5){$\scriptstyle1$};\node[below]at(5.8,-1.5){$\scriptstyle \infty$};\node[below]at(3.5,-1.5){$\scriptstyle{H_D}$};

\foreach \point in {Tinterrupt1,Tinterrupt2,Qinterrupt1,Qinterrupt2}{\draw (\point)--++(-.05,-.1);\draw (\point)--++(.05,.1);}

\path (5.5,0) coordinate (tp1); \path (4,0) coordinate (tp2); \path (2.5,0) coordinate (tp3a); \path (1.5,.75) coordinate (tp3b);
\draw [->] (T)..controls+(15:1cm) and +(-90:1cm)..(tp1); \draw (tp1)..controls+(90:1cm) and +(-30:.5cm)..(S);
\draw [->] (T)..controls +(45:.5cm) and +(-90:.7cm)..(tp2); \draw (tp2)..controls +(90:.7cm) and +(-90:.5cm)..(S);
\draw [->] (T)..controls +(60:.5cm) and +(-90:.5cm)..(tp3a); \draw (tp3a)..controls +(90:.5cm) and +(-90:.2cm)..(tp3b);
\draw (tp3b)..controls +(90:.2cm) and +(-120:.1cm)..(S);

\end{tikzpicture}
\caption{The flow in $\overline{\mathcal V}_\mathrm{VIII}$.}\label{F:flow on V}
\end{figure}

\subsection{Analysis in $\overline{\mathcal B}_\mathrm{III}$}\label{SS:analysis on B}

The dynamical system in $\overline{\mathcal B}_\mathrm{III}$ is obtained from~\eqref{E:dsX} by setting $M_1=0$ $(\Leftrightarrow s=0)$ and using~\eqref{E:beta}:
\begin{equation}\label{E:dsB}
\begin{bmatrix} H_D \\ \Sigma_+ \end{bmatrix}' =
\begin{bmatrix} -(1-H_D^2)\left(2-\frac{3}{2}(1-w)(1-\Sigma_+^2)-H_D\Sigma_+\right) \\ -(1-\Sigma_+^2)\left((1-H_D^2)+\frac{3}{2}(1-w)(H_D\Sigma_++\beta)\right) \end{bmatrix}
=: f(H_D,\Sigma_+)
\end{equation}
This system has at least three and at most five fixed points in $\overline{\mathcal B}_\mathrm{III}$ depending on $(w,\beta)\in\mathbb P$, namely
\begin{equation}\notag
T_\flat:=\begin{bmatrix} 1 \\ -1 \end{bmatrix},
Q_\flat:=\begin{bmatrix} 1 \\ 1 \end{bmatrix},
D=\begin{bmatrix} 2 \\Ê1 \end{bmatrix},
R_\flat:=\begin{bmatrix} 1 \\ -\beta \end{bmatrix} \text{and }
P:=\begin{bmatrix} \frac{2+3\beta(1-w)}{\sqrt{(1-3w)^2+6\beta(1-w)}} \\ \frac{1+3w}{\sqrt{(1-3w)^2+6\beta(1-w)}} \end{bmatrix}.
\end{equation}
Clearly $R_\flat$ is a fixed point in $\overline{\mathcal B}_\mathrm{III}$ (and different from $T_\flat$, $Q_\flat$) iff $\beta\in(-1,1)$. The conditions $H_D|_P>1$ and $\Sigma_+|_P<1$ entail that $P$ is in $\mathcal B_\mathrm{III}$ iff $\beta > \beta_\flat(w) := \frac{2w}{1-w}$ and $(w,\beta)\not\in\overline{\mathbb P}_\supset$, where $\mathbb P_\supset$ refers to the small $\supset$-shaped subset of $\mathbb P$ bounded by $\beta_\pm(w):=\frac{-1\pm\sqrt{-3+(1-3w)^2}}{3(1-w)}$; cf. figure~\ref{F:parameter space}. Under these conditions the square root in the coordinates of $P$ is automatically real and $\Sigma_+|_P>0$.

The eigenvectors and eigenvalues of the linearisation $\mathrm Df(H_D,\Sigma_+)$ at $T_\flat, Q_\flat, D$ and $R_\flat$ determine their local stability properties. One finds
\begin{align*}\notag
T_\flat: \begin{bmatrix} 6 \\ \scriptstyle{3(1-w)(1-\beta)} \end{bmatrix} &\begin{bmatrix} 1 & 0 \\ 0 & 1 \end{bmatrix}, \quad
&&D\,:\begin{bmatrix} -3 \\ 3\beta(1-w)-6w \end{bmatrix}~\begin{bmatrix} 1 & \frac{1-3w}{1-2w+\beta(1-w)} \\ 0 & 1 \end{bmatrix}, \\
Q_\flat: \begin{bmatrix} 2 \\ \scriptstyle{3(1-w)(1+\beta)} \end{bmatrix} &\begin{bmatrix} 1 & 0 \\ 0 & 1 \end{bmatrix}, \quad
&&R_\flat:\begin{bmatrix} \scriptstyle{3\beta^2(1-w)+1+3w+2\beta} \\ -\frac{3}{2}(1-w)(1-\beta^2) \end{bmatrix} \begin{bmatrix} \frac{3\beta^2(1-w)+5+3w+4\beta}{(1-\beta^2)(3\beta(1-w)+4)} & 0 \\ 1 & 1 \end{bmatrix}.
\end{align*}
Therefore, in $\overline{\mathcal B}_\mathrm{III}$,  $T_\flat$ is a source for $\beta<1$ and a saddle repelling in $H_D$ direction for $\beta>1$, $Q_\flat$ is a saddle repelling in $H_D$ direction for $\beta<-1$ and a source for $\beta>-1$, $D$ is a sink for $\beta<\beta_\flat(w)$ and a saddle attracting in $H_D$ direction for $\beta>\beta_\flat(w)$ and $R_\flat$ is a sink $\forall (w,\beta)\in\mathbb P_\supset$ and a saddle attracting in $\Sigma_+$ direction $\forall (w,-1<\beta<1)\in\mathbb P\backslash\overline{\mathbb P}_\supset$.

The eigenvalues $\lambda_\pm$ of $\mathrm Df|_P$ are given in terms of trace and determinant by
\begin{equation}\label{E:eigenvalues(P)}
2\lambda_\pm = \tr\mathrm Df|_P \pm\sqrt{(\tr\mathrm Df|_P)^2-4\det\mathrm Df|_P},
\end{equation}
where
\begin{align}
\det\mathrm Df|_P &= \frac{9(1-w)\big(3\beta^2(1-w)+1+3w+2\beta\big)\big(\beta(1-w)-2w\big)}{(1-3w)^2+6\beta(1-w)}, \label{E:det(Dfb)} \\
\tr\mathrm Df|_P &= -\frac{3(1-w)(1+2\beta)}{\sqrt{(1-3w)^2+6\beta(1-w)}}, \label{E:tr(Dfb)}
\end{align}
which are real whenever $P$ exists in $\mathcal B_\mathrm{III}$ because the square root in~\eqref{E:tr(Dfb)} appears in the coordinates of $P$ as well. Further, $\tr\mathrm Df|_P<0$ and $\det\mathrm Df|_P>0$ whenever $P$ is in $\mathcal B_\mathrm{III}$.\footnote{To see the first inequality, note from figure~\ref{F:parameter space} that $\beta>-\frac{1}{2}$ when $P$ is in $\mathcal B_\mathrm{III}$. To see the second inequality, note that setting the middle and last factor in the numerator of~\eqref{E:det(Dfb)} to zero corresponds to the bifurcation line $\beta_\pm(w)$ and $\beta_\flat(w)$, respectively.} Hence the real part of~\eqref{E:eigenvalues(P)} is always negative, whether the eigenvalues are real or complex, so $P$ is a sink in $\mathcal B_\mathrm{III}$ whenever present. The eigenvalues are real for $\beta\in(\beta_\flat(w),\beta_\flat^*(w)]$; cf. figure~\ref{F:parameter space}.

In order to perform the analysis of the flow in the full state space $\mathcal X_\mathrm{VIII}$ in section~\ref{SS:Analysis on X}, it is also necessary to know the local stability of $P$ in the direction of $M_1$: From~\eqref{E:dsX} one has $(\ln M_1)'|_P=-3\Sigma_+|_P$, which is negative so that $P$ is attracting in the direction of $M_1$. This implies that $P$ is a local sink in $\overline{\mathcal X}_\mathrm{VIII}$.

It is shown in appendix~\ref{SS:appendix2dim} that there are exactly two more fixed points at infinity, $T_\infty=[\infty,-1]^\mathrm T$ and $Q_\infty=[\infty,1]^\mathrm T$. Their stability properties then follow straightforwardly from~\eqref{E:dsB}: For $H_D$ sufficiently large, $f(H_D,\Sigma_+)\approx[-H_D^3\Sigma_+,\allowbreak H_D^2(1-\Sigma_+^2)]^\mathrm T$. Hence, $\Sigma_+'>0$ and $H_D' \gtreqqless 0$ for $\Sigma_+ \lesseqqgtr 0$, which implies that $T_\infty$ and $Q_\infty$ play the role of saddles for the flow in $\overline{\mathcal B}_\mathrm{III}$.

For the cases where $P$ does not exist in $\mathcal B_\mathrm{III}$, the information suffices to draw the corresponding qualitative flow diagrams. When $P$ exists in $\mathcal B_\mathrm{III}$, periodic orbits encircling this fixed point could in principle be present. However numeric investigations strongly suggest that this is not the case. In any case, the main results of this paper are not affected by this open question; cf. section~\ref{S:Results}.

The resulting qualitative dynamics in $\overline{\mathcal B}_\mathrm{III}$ in dependence of $(w,\beta)\in\mathbb P$ is depicted in figures~\ref{F:flow13} to~\ref{F:flow1} as flows in $\overline{\mathcal B}_\mathrm{III}^\mathcal Y$ together with the flows in $\overline{\mathcal S}_\sharp^\mathcal Y$ and $\overline{\mathcal X}_\mathrm I$. The latter two subsets will be discussed next.

\subsection{Analysis in $\overline{\mathcal S}_\sharp$}\label{SS:analysis on S}
$\overline{\mathcal S}_\sharp$ is the common boundary of the state spaces $\mathcal X_\mathrm{VIII}$ and $\mathcal X_\mathrm{IX}$ of LRS types~VIII and~IX. It has been analysed in detail in~\cite[section 10.1]{M&S}. The bifurcation lines $\beta=\pm2$, $\beta_\sharp(w)$ and $\beta_\sharp^*(w)$ result from this analysis. The results are depicted in section~\ref{S:Results} as flows in $\overline{\mathcal S}_\sharp^\mathcal Y$ together with the flows in $\overline{\mathcal B}_\mathrm{III}^\mathcal Y$ and $\overline{\mathcal X}_\mathrm I$.

\subsection{Analysis in $\overline{\mathcal X}_\mathrm I$}\label{SS:analysis on XI}

The flow in $\overline{\mathcal X}_\mathrm I$ has been analysed in detail in~\cite[section 10.2]{M&S}. The only new bifurcation line in figure~\ref{F:parameter space} resulting from this analysis is $\beta=0$ which corresponds to the local stability properties of the Friedmann point $F$. However the qualitative dynamics in $\overline{\mathcal X}_\mathrm I$ is also associated with the bifurcation lines $\beta=\pm1$ and $\beta=\pm2$. The results are depicted in section~\ref{S:Results} together with the flows in $\overline{\mathcal B}_\mathrm{III}^\mathcal Y$ and $\overline{\mathcal S}_\sharp^\mathcal Y$.

\subsection{Analysis in $\overline{\mathcal X}_\mathrm{VIII}$ and $\overline{\mathcal S}_\infty$}\label{SS:Analysis on X}

The full system~\eqref{E:dsX} does not have any fixed points in $\mathcal X_\mathrm{VIII}$. Furthermore, as shown in appendix~\ref{SS:appendix3dim} there are only the already know fixed points $T_\infty=[\infty,-1,0]^\mathrm T$ and $Q_\infty=[\infty,1,0]^\mathrm T$ at infinity. Their stability properties then follow straightforwardly from~\eqref{E:dsX}: For $H_D$ sufficiently large, $\mathrm{rhs}\eqref{E:dsX}\approx[-H_D^3\Sigma_+,H_D^2(1-\Sigma_+^2),-H_D^2\Sigma_+M_1]^\mathrm T$. Hence $\Sigma_+'>0$, $H_D'\gtreqqless0$ for $\Sigma_+\lesseqqgtr0$ and $M_1'\gtreqqless0$ for $\Sigma_+\lesseqqgtr0$, which implies that $T_\infty$ and $Q_\infty$ play the role of saddles for the flow in $\overline{\mathcal X}_\mathrm{VIII}$. Viewing infinity as the boundary $\overline{\mathcal S}_\infty$ in a compactified version of the state space, $T_\infty$ is a source and $Q_\infty$ is a sink in $\overline{\mathcal S}_\infty$, hence all orbits in $\overline{\mathcal S}_\infty$ are of the form $T_\infty\to Q_\infty$.

An important consequence is that there is no orbit in $\mathcal X_\mathrm{VIII}$ that emanates from or escapes to infinity. Hence each orbit in $\mathcal X_\mathrm{VIII}$ lies in a compact subset. The following lemma concludes this section. It will also be used to localise the $\alpha$- and $\omega$-limit sets~\cite[p~99 Def~4.12]{WE} in the next section; cf.~\cite[p~91 Def~4.7]{WE} for the term `invariant set':

\begin{lemma}\label{L:monotonicity principle} Let $\gamma$ be an orbit in $\mathcal Y_\mathrm{VIII}$. Then both, $\alpha(\gamma)$ and $\omega(\gamma)$, is a non-empty, compact and connected invariant subset of $\overline{\mathcal Y}_\mathrm{VIII}$. Furthermore, $\alpha(\gamma)\subseteq\overline{\mathcal S}_\sharp^\mathcal Y\cup\overline{\mathcal X}_\mathrm I$ and $\omega(\gamma)\subseteq\overline{\mathcal B}_\mathrm{III}^\mathcal Y\cup\overline{\mathcal X}_\mathrm I$.
\end{lemma}
\begin{proof}
First, since each orbit in $\mathcal X_\mathrm{VIII}$ lies in a compact subset, so does by~\eqref{E:coord trafo} each orbit in $\mathcal Y_\mathrm{VIII}$. The first statement of the lemma then follows from~\cite[p~99 Thm~4.9]{WE}.

Second, consider the function $Z_5:\mathcal X_\mathrm{VIII}\cup\mathcal V_\mathrm{VIII}\to\mathbb R$ given by
\begin{equation}\notag
Z_5:=\frac{H_DM_1^\frac{1}{3}}{(H_D^2-1)^\frac{2}{3}}>0\quad\text{with}\quad Z_5'=-\frac{4\Sigma_++(1+3w)\Omega}{2H_D}Z_5\le0.
\end{equation}
One can check that $Z_5''|_{\Sigma_+=\Omega=0}=0$ and $Z_5'''|_{\Sigma_+=\Omega=0}<0$ in $\mathcal X_\mathrm{VIII}$, which means that $Z_5$ is strictly monotonically decreasing along the flow in $\mathcal X_\mathrm{VIII}$. Hence the monotonicity principle \cite[p~103 Thm~4.12]{WE} implies that the limit sets lie in $\overline{\mathcal B}_\mathrm{III}\cup\overline{\mathcal S}_\sharp$. Moreover, since $Z_5\to0=\mathrm{inf}(Z_5)$ when one approaches $\mathcal B_\mathrm{III}$ and $Z_5\to\infty=\mathrm{sup}(Z_5)$ when one approaches $\mathcal S_\sharp$, the monotonicity principle implies $\alpha(\tilde\gamma)\subseteq\overline{\mathcal S}_\sharp$ and $\omega(\tilde\gamma)\subseteq\overline{\mathcal B}_\mathrm{III}$ for all orbits $\tilde\gamma$ in $\mathcal X_\mathrm{VIII}$. In the context of $\mathcal Y_\mathrm{VIII}$ this means that $\alpha(\gamma)\subseteq\overline{\mathcal S}_\sharp^\mathcal Y\cup\overline{\mathcal X}_\mathrm I$ and $\omega(\gamma)\subseteq\overline{\mathcal B}_\mathrm{III}^\mathcal Y\cup\overline{\mathcal X}_\mathrm I$; cf.~\eqref{E:correspondences}.
\end{proof}


\section{Results and discussion}\label{S:Results}

Finally all information obtained in section~\ref{S:ds analysis} can be collected to identify the $\alpha$- and $\omega$-limit sets of generic orbits $\gamma\in\mathcal Y_\mathrm{VIII}$, and hence the past and future asymptotic dynamics of generic LRS Bianchi type~VIII cosmologies, for all qualitatively different anisotropic matter cases:

\subsection{Identification of the limit sets}\label{SS:limit sets}

Figure~\ref{F:parameter space} shows that there are thirteen qualitatively different cases.\footnote{The bifurcation lines $\beta_\flat^*(w)$ and $\beta_\sharp^*(w)$ only specify if the fixed points $C_\sharp$ and $P$ are local stable nodes or foci in $\mathcal S_\sharp^\mathcal Y$ and $\mathcal B_\mathrm{III}^\mathcal Y$ respectively, which is irrelevant for identifying the limit sets.} These cases are listed in figures~\ref{F:flow13} to~\ref{F:flow1}, which show, for each case, a plot of the bifurcation diagram where the corresponding subset of $\mathbb P$ is shaded and a representative flow diagram in the set $\overline{\mathcal S}_\sharp^\mathcal Y\cup\overline{\mathcal X}_\mathrm I\cup\overline{\mathcal B}_\mathrm{III}^\mathcal Y$, on which the limit sets lie by lemma~\ref{L:monotonicity principle}. Solid (dashed) lines represent generic (non-generic) orbits in the respective boundary subset. A filled (empty) circle on a fixed point indicates that it attracts (repels) orbits in the respective orthogonal direction.\footnote{The colour coding of $T_\infty$ and $Q_\infty$ makes only sense in the context of $\mathcal X_\mathrm{VIII}$. The fixed points in $\partial\mathcal X_\mathrm I$ are not colour coded since the orthogonal directions are represented in the figures.} Heteroclinic cycles and networks are represented by thick lines. The axes with two labels emphasise the diffeomorphisms~\eqref{E:isomorphisms} and the bijection between $s$ and $\theta$; cf. section~\ref{SS:state space Y}.

By lemma~\ref{L:monotonicity principle}, the limit sets are non-empty, compact and connected invariant subsets of $\overline{\mathcal Y}_\mathrm{VIII}^\mathcal Y$. This leaves fixed points, heteroclinic cycles and heteroclinic networks as candidates for the limit sets in figures~\ref{F:flow13} to~\ref{F:flow1}.

A single fixed point is an ($\alpha$-) $\omega$-limit set of generic orbits $\gamma\in\mathcal Y_\mathrm{VIII}$ iff it (repels) attracts orbits from a three dimensional neighbourhood in $\mathcal Y_\mathrm{VIII}$; e.g. a (source) sink. These can easily be identified from the figures. In the cases of figures~\ref{F:flow9} to~\ref{F:flow2}, fixed points are the only possible generic limit sets. Furthermore, these cases exhibit just one source and sink respectively, which implies that these are the past/future attractors~\cite[p~100 Def~4.13]{WE}.

Figures~\ref{F:flow13} to~\ref{F:flow10} and~\ref{F:flow1} show a heteroclinic cycle or network each, and hence additional candidates for limit sets. In figures~\ref{F:flow11} and~\ref{F:flow10} the cycle is the past attractor since there is no other candidate for a generic $\alpha$-limit set in $\overline{\mathcal S}_\sharp^\mathcal Y\cup\overline{\mathcal X}_\mathrm{I}$, and as proven in section~\ref{SS:Analysis on X} $\alpha(\gamma)$ is non-empty. Note that since the cycle is not in $\overline{\mathcal B}_\mathrm{III}^\mathcal Y\cup\overline{\mathcal X}_\mathrm{I}$ it cannot be an $\omega$-limit set. The same line of arguments shows that in figures~\ref{F:flow13} and~\ref{F:flow12} the past attractor is a subset of the heteroclinic network.

In figure~\ref{F:flow1}, the heteroclinic network could be an $\alpha$-limit set in addition to $T_\flat$. Analogously, in figures~\ref{F:flow13} and~\ref{F:flow12}, the heteroclinic cycle $\partial\mathcal X_\mathrm I$ could be an $\omega$-limit set in addition to $D$ and $P$, respectively. The nontrivial task of directly tackling the stability properties of these structures is omitted here. They are however taken into account as further candidates for generic limit sets in these three cases. Note that figure~\ref{F:flow1} is the only one of these three cases where the energy conditions can be satisfied; cf. section~\ref{SS:matter family} and figure~\ref{F:parameter space}.

\subsection{The main results}

Finally one arrives at the main theorem and its corollaries:

\begin{theorem}\label{T:limit sets}
The $\alpha$- and $\omega$-limit sets of generic orbits $\gamma\in\mathcal Y_\mathrm{VIII}$ of the dynamical system~\eqref{E:dsX} with parameters $(w,\beta)\in\mathbb P$, is given as stated in the captions of figures~\emph{\ref{F:flow13}} to~\emph{\ref{F:flow1}}. These describe the past and future asymptotic dynamics of LRS Bianchi type~VIII cosmologies with anisotropic matter of the family defined in section~\emph{\ref{SS:matter family}}.
\end{theorem}
\begin{proof}
Cf. section~\ref{SS:limit sets}.
\end{proof}

The exact solutions corresponding to the fixed points are summarised in table~\ref{Table:exact solutions}. These have been given in~\cite[appendix~A]{M&S} for all fixed points but $D$. For $D$, one just needs to insert the coordinates into~\cite[Eq~99]{M&S}. The corresponding isotropic cases to these solutions can be found in~\cite[section~9.1]{WE}.

\begin{table}\begin{tabular}{ c c c c}
Fixed points & corresponding solution & type & vacuum\\ \hline
$T_\flat,T_\sharp$ & Taub Kasner & I & \checkmark \\
$Q_\flat,Q_\sharp$ & non-flat LRS Kasner & I & \checkmark \\ 
$R_\flat,R_\sharp$ & a type I anisotropic matter solution & I & \\
$F$ & flat Friedmann & I & \\
$C_\sharp$ & generalisation of Collins-Stewart & II & \\
$P$ & generalisation of Collins $(\mathrm{VI}_{-1})$\phantom{ } & III & \\
$D$ & type~III form of flat spacetime & III & \checkmark \\
& & & \\
\end{tabular}
\caption{The exact solutions corresponding to the fixed points.}\label{Table:exact solutions}
\end{table}

\begin{perfect fluid}\label{C:vac and perf fluid}
LRS Bianchi type~VIII vacuum solutions are past asymptotic to $T_\flat$ and future asymptotic to $D$.

Generic LRS Bianchi type~VIII solutions with a non-tilted perfect fluid where $p=w\rho$ and $w\in\left(-\frac{1}{3},1\right)$ are past asymptotic to $T_\flat$, future asymptotic to $P$ for $w\in\left(-\frac{1}{3},0\right)$ and to $D$ for $w\in[0,1)$.
\end{perfect fluid}

The vacuum part follows from section~\ref{SS:analysis on V}; cf. figure~\ref{F:flow on V}. For the perfect fluid cases, recall from section~\ref{SS:matter family} that they correspond to the line $\beta=0$ in $\mathbb P$. They are therefore contained in theorem~\ref{T:limit sets} as special cases.\footnote{The fact that $\beta=0$ is itself a bifurcation line, which is related to the stability of the Friedmann point, is not relevant for the asymptotic dynamics.} The only bifurcation line intersecting $\beta=0$ is $\beta_\flat(w)$ at $w=0$ which yields the two qualitatively different perfect fluid cases, see figures~\ref{F:flow6} and~\ref{F:flow5}, respectively. The statement on the future asymptotics for $w\in\left(-\frac{1}{3},0\right)$ might fill a little gap in the literature. The other results are known: The past asymptotics has been given in~\cite[table~4 and figure~4]{Wainwright&Hsu}, the future asymptotics for vacuum in~\cite[Prop~8.1]{Ringstroem} and the future asymptotics for perfect fluids with $w\in[0,1)$ in~\cite[Thm~3.1]{Horwood et al}. 

Comparison with the cases of theorem~\ref{T:limit sets} shows that the dynamics with anisotropic matter differs from the vacuum and perfect fluid cases in figures~\ref{F:flow13} to~\ref{F:flow10},~\ref{F:flow4} and perhaps~\ref{F:flow1}. The strongest implications of this represent the main results of this paper stated in the following corollaries:
\begin{corollary}[\textbf{past asymptotics}]\label{C:past asymptotics}
The past asymptotic dynamics of generic LRS Bianchi type~VIII cosmologies with anisotropic matter can differ significantly from that of the vacuum and perfect fluid cases. 
In particular, the approach to the initial singularity is oscillatory in the cases depicted in figures~\ref{F:flow13} to~\ref{F:flow10} and perhaps~\ref{F:flow1}.
\end{corollary}
\begin{corollary}[\textbf{future asymptotics}]\label{C:future asymptotics}
The future asymptotic dynamics of generic LRS Bianchi type~VIII cosmologies with anisotropic matter can differ significantly from that of the vacuum and perfect fluid cases.
Note in particular the cases of figures~\ref{F:flow4} and perhaps~\ref{F:flow13} and~\ref{F:flow12}.
\end{corollary}

For completeness it is necessary to elaborate on the remark at the end of section~\ref{SS:analysis on B} concerning the possible occurrence of periodic orbits around $P$ in $\mathcal B_\mathrm{III}^\mathcal Y$: Such orbits are not observed numerically. However, their occurrence could change certain details in theorem~\ref{T:limit sets}, adding further $\omega$-limit candidates in the cases where $P$ is in $\mathcal B_\mathrm{III}^\mathcal Y$. On the other hand, corollary~\ref{C:past asymptotics} would be completely untouched and corollary~\ref{C:future asymptotics} even strengthened by the occurrence of such periodic orbits.

\subsection{Physical interpretation of the results}\label{SS:discussion}

An interesting observation is that there is a neighbourhood of the line $\beta=0$ in $\mathbb P$ for which the corresponding asymptotic dynamics is identical to that of the perfect fluid cases. Hence the perfect fluid solutions are robust under small perturbations of a vanishing anisotropy parameter. For the past asymptotics to differ from that, $\beta$ even needs to be so high that the dominant energy condition is only marginally satisfied.

The specialisation of the results to the cases for which the energy conditions are satisfied is obtained by restricting $\beta$ to $\bigl[\mathrm{max}\bigl(-2,-\frac{1+w}{1-w}\bigr),1\bigr]$; cf. section~\ref{SS:matter family}. This corresponds to the shaded region in figure~\ref{F:parameter space}, and thus rules out the cases of figures~\ref{F:flow13} and~\ref{F:flow12}. All other cases contain values $(w,\beta)$ for which the energy conditions are satisfied.

However, it should be pointed out that the cases represented by figures~\ref{F:flow11} and~\ref{F:flow10} only satisfy the energy conditions for $\beta=1$, and the case of figure~\ref{F:flow1} only for $w\geq\frac{1}{3}$ and $\beta=-2$, i.e. for $(w,\beta)$ in one-dimensional subsets of $\mathbb P$ for which the dominant energy condition is only marginally satisfied. Consequently, in these cases the energy flow is necessarily lightlike close to the singularity. A special example is given by collisionless (Vlasov) matter with massless particles, which as shown in~\cite[section~12.1]{M&S} falls into the class of matter models with $(w,\beta)=\left(\frac{1}{3},1\right)$.

In any case, the statements of corollaries~\ref{C:past asymptotics} and~\ref{C:future asymptotics} remain true under the restriction to matter satisfying the energy conditions.\\

This concludes the main results on the anisotropic matter analysis. Section~\ref{S:Vlasov} is concerned with an extension of the present formalism to treat LRS type~VIII dynamics with Vlasov matter with massive particles, which does not a priori fit into the anisotropic matter family considered so far.

\begin{figure}[H]

\caption{$\alpha(\gamma)\subseteq\mathrm{network}(T_\flat,T_\sharp,Q_\sharp,Q_\flat)$. $\omega(\gamma)=D$, possibly $\partial\mathcal X_\mathrm I$.}
\label{F:flow13}
\end{figure} 

\begin{figure}[H]%
\caption{$\alpha(\gamma)\subseteq\mathrm{network}(T_\flat,T_\sharp,Q_\sharp,Q_\flat)$. $\omega(\gamma)=D$, possibly $\partial\mathcal X_\mathrm I$.}
\label{F:flow12}
\end{figure}

\begin{figure}[H]%
\caption{$\alpha(\gamma)=\mathrm{cycle}(T_\flat,T_\sharp,Q_\sharp,Q_\flat)$. $\omega(\gamma)=P$.}
\label{F:flow11}
\end{figure}

\begin{figure}[H]%
\caption{$\alpha(\gamma)=\mathrm{cycle}(T_\flat,T_\sharp,Q_\sharp,Q_\flat)$. $\omega(\gamma)=D$.}
\label{F:flow10}
\end{figure}

\begin{figure}[H]%
\caption{$\alpha(\gamma)=T_\flat$. $\omega(\gamma)=D$.}\label{F:flow9}
\end{figure}

\begin{figure}[H]%
\caption{$\alpha(\gamma)=T_\flat$. $\omega(\gamma)=P$.}\label{F:flow8}
\end{figure}

\begin{figure}[H]%
\caption{$\alpha(\gamma)=T_\flat$. $\omega(\gamma)=P$.}\label{F:flow7}
\end{figure}

\begin{figure}[H]%
\caption{$\alpha(\gamma)=T_\flat$. $\omega(\gamma)=D$.}\label{F:flow6}
\end{figure}

\begin{figure}[H]%
\caption{$\alpha(\gamma)=T_\flat$. $\omega(\gamma)=P$.}\label{F:flow5}
\end{figure}

\begin{figure}[H]
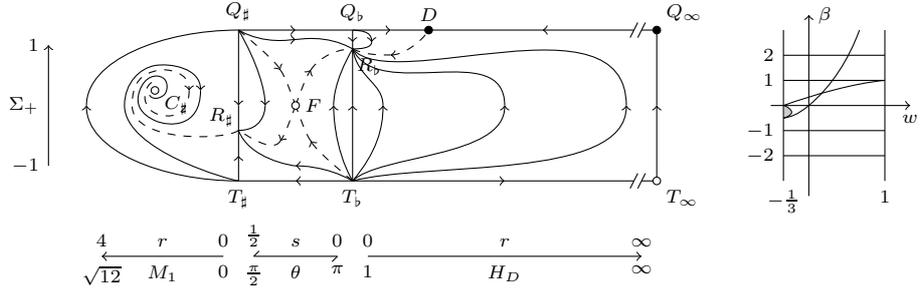
%
\caption{$\alpha(\gamma)=T_\flat$. $\omega(\gamma)=R_\flat$.}\label{F:flow4}
\end{figure} 

\begin{figure}[H]%
\caption{$\alpha(\gamma)=T_\flat$. $\omega(\gamma)=D$.}\label{F:flow3}
\end{figure}

\begin{figure}[H]%
\caption{$\alpha(\gamma)=T_\flat$, possibly $\mathrm{network}(T_\sharp,Q_\sharp,F)$. $\omega(\gamma)=D$.}\label{F:flow1}
\end{figure}


\section{Extension of the formalism to treat Vlasov matter dynamics with massive particles}\label{S:Vlasov}

It was stated in section~\ref{SS:discussion} that Vlasov matter with massless particles falls into the class of matter models described in section~\ref{SS:matter family} with parameters $(w,\beta)=(\frac{1}{3},1)$. This is shown in~\cite[section~12.1]{M&S}. However, as stated there, for Vlasov matter with massive particles the relation $p=w\rho$ is non-linear. In~\cite{M&S;Vlasov}, \mbox{Calogero} and \mbox{Heinzle} extended their formalism to treat Vlasov matter dynamics with massive particles in the case of LRS type~IX. Since types~VIII and~IX are analogous in this formulation, and in particular share the same evolution equations, this extension of the formalism can be adopted to LRS type~VIII without difficulties:

Following~\cite[section~3]{M&S;Vlasov}, the dynamical system representing the LRS type~VIII Einstein-Vlasov system is given by the system~\eqref{E:dsX} and the additional equation
\begin{equation}\label{E:l'}
l'=2H_Dl(1-l).
\end{equation}
The additional variable $l:=\frac{(\det g)^{1/3}}{1+(\det g)^{1/3}}\in(0,1)$ corresponds to a length scale of the spatial metric. The rescaled principal pressures are functions of $l$ and $s$. Their full expressions are given in~\cite[Eq~13]{M&S;Vlasov} from which it follows that
\begin{equation}\label{E:(w,b)}
(w,\beta)|_{l=0}=\left(\textstyle{\frac{1}{3}},1\right) \qquad\text{and}\qquad (w,\beta)|_{l=1}=(0,0).
\end{equation}
The LRS type~VIII state space for Vlasov matter dynamics is given by $\mathcal X_\mathrm{VIII}\times(0,1)$.

The dynamical system can be analysed as follows: By~\eqref{E:l'}, $l$ is strictly monotonically increasing along orbits in $\overline{\mathcal X}_\mathrm{VIII}\times(0,1)$. The monotonicity principle thus implies that $\alpha(\gamma)\subseteq\overline{\mathcal X}_\mathrm{VIII}\times\{0\}$ and $\omega(\gamma)\subseteq\overline{\mathcal X}_\mathrm{VIII}\times\{1\}$.\footnote{One can use the same arguments as in section~\ref{SS:Analysis on X} and appendix~\ref{S:appendix} for each $l=\mathrm{const}$ hypersurface to prove that all orbits can be trapped in a compact subset of $\mathbb R^4$.} Hence, the search for the limit sets can be restricted to these boundary subsets:

From~\eqref{E:(w,b)}, the flow in $\overline{\mathcal X}_\mathrm{VIII}\times\{0\}$ is equivalent to the flow for anisotropic matter with parameters $(\frac{1}{3},1)$; cf. figure~\ref{F:flow10}. Hence the $\mathrm{cycle}(T_\flat,T_\sharp,Q_\sharp,Q_\flat)$ is the past attractor of the flow in $\mathcal X_\mathrm{VIII}\times(0,1)$, since this cycle is the only candidate for an $\alpha$-limit set for generic orbits. Also from~\eqref{E:(w,b)}, the flow in $\overline{\mathcal X}_\mathrm{VIII}\times\{1\}$ is equivalent to the flow for anisotropic matter with parameters $(0,0)$, i.e. to that for dust; cf. figure~\ref{F:flow6}. Hence $D$ is the future attractor of the flow in $\mathcal X_\mathrm{VIII}\times(0,1)$, since this fixed point is the only candidate for an $\omega$-limit set for generic orbits.

$l=0$ corresponds to Vlasov dynamics with massless particles; cf.~\cite[section~3]{M&S;Vlasov}. Hence the result states that Vlasov matter with massive particles behaves like Vlasov matter with massless particles asymptotically to the past, and like dust asymptotically to the future.


\begin{appendix}

\section{Fixed points at infinity}\label{S:appendix}

To find the fixed points of~\eqref{E:dsX}, \eqref{E:dsV} and \eqref{E:dsB} at infinity, a method presented in \cite[3.10]{Perko} is adopted in a slightly modified way:

\subsection{In two dimensions}\label{SS:appendix2dim}

Consider a two dimensional system
\begin{equation*}
\begin{bmatrix} H_D \\ \Sigma_+ \end{bmatrix}' = \begin{bmatrix} P(H_D,\Sigma_+) \\ Q(H_D,\Sigma_+) \end{bmatrix}
\end{equation*}
where $P$ and $Q$ are polynomials of degree $n$ and $m$ in $H_D$ respectively, satisfying $n\leq m+1$. One can write this system as
\begin{equation}\label{E:ds2}
Q(H_D,\Sigma_+)\mathrm dH_D-P(H_D,\Sigma_+)\mathrm d\Sigma_+=0 ,
\end{equation}
where however the information about the direction of the flow is lost. Next the $H_D$~coordinate is formally compactified by projecting the state space onto the `Poincar\'e cylinder' as illustrated in figure~\ref{F:Poincare cylinder}. The corresponding coordinate transformation $(H_D,\Sigma_+)\rightarrow(X,\Sigma_+,Z):X^2+Z^2=1$ is given by $H_D=\frac{X}{Z}$. Applying this to \eqref{E:ds2} and multiplying by $Z^{m+2}$ yields
\begin{equation}\label{E:ds2XYZ}
Z^{m+1}Q(X/Z,\Sigma_+)\mathrm dX-Z^{m+2}P(X/Z,\Sigma_+)\mathrm d\Sigma_+-XZ^mQ(X/Z,\Sigma_+)\mathrm dZ=0,
\end{equation}
which defines a flow on the `Poincar\'e cylinder'; cf. \cite[p~267]{Perko}. Points at infinity in the original state space correspond to points on the `equator' $X=1,Z=0$ on the `Poincar\'e cylinder'. Furthermore, evaluating~\eqref{E:ds2XYZ} with $X=1,Z=0$ gives the flow on the `equator', which corresponds to the flow at infinity,
\begin{equation}\label{E:ds2equator}
\bigl(Z^mQ(X/Z,\Sigma_+)\bigr)\big|_{X=1, Z=0}\mathrm dZ=0.
\end{equation}
Note that the first two terms of~\eqref{E:ds2XYZ} vanish since they are at least proportional to $Z$. This is not true however for the third term of~\eqref{E:ds2XYZ}. From~\eqref{E:ds2equator}, for $(Z^mQ)|_{X=1,Z=0}\neq0$ it follows that $\mathrm dZ=0$, which corresponds to trajectories through a regular point on the `equator', where the sign of $(Z^mQ)|_{X=1,Z=0}$ determines the flow direction. Fixed points on the `equator' correspond to solutions of $(Z^mQ)|_{X=1,Z=0}=0$; cf. \cite[p~268 Thm~1]{Perko}. Solving this equation for $Q$ in the context of~\eqref{E:dsV} and~\eqref{E:dsB} yields $\Sigma_+=\pm1$ for the fixed points at infinity in $\mathcal V_\mathrm{VIII}$ and $\mathcal B_\mathrm{III}$ respectively, i.e. $T_\infty:=[\infty,-1]^\mathrm T$ and $Q_\infty:=[\infty,1]^\mathrm T$.
\begin{figure}\begin{tikzpicture}
\draw[->](-2.5,0)--(2.5,0);		\node[below right]at(2.5,0){$\scriptstyle X$};			
\draw[->](0,-.6)--(0,1.8);		\node[above left]at(0,1.8){$\scriptstyle Z$};				
\draw[->](-2.5,1.5)--(2.5,1.5);	\node[below right]at(2.5,1.5){$\scriptstyle{H_D}$}; 	
\draw(1.5,0)arc(0:180:1.5cm);					%
\draw[dashed](1.5,0)arc(0:-20:1.5cm);		
\draw[dashed](-1.5,0) arc (180:200:1.5cm);	%
\node[below right]at(1.5,0){$\scriptstyle 1$};
\draw(0,0)--(2,1.5);									
\draw[fill] (2,1.5) circle (0.025);	\node[above]at(2,1.5){$\scriptstyle{(H_D,\Sigma_+)}$};	
\draw[fill] (37:1.5) circle (0.025);\node[right]at(35:1.5){$\scriptstyle{(X,\Sigma_+,Z)}$};	
\draw(37:1.5cm)--++(-1.2,0);\draw(1.2,-.05)--(1.2,.05);
\end{tikzpicture}
\caption{Projection onto the `Poincar\'e cylinder'.}\label{F:Poincare cylinder}\end{figure}
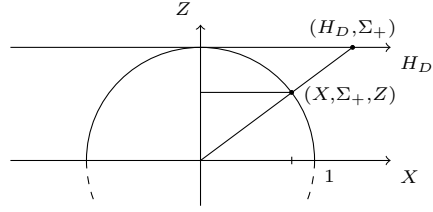

\subsection{In three dimensions}\label{SS:appendix3dim}

The generalisation to three or more dimensions is straightforward; cf. \cite[p~277~ff]{Perko}. Consider a three dimensional system
\begin{equation*}
\begin{bmatrix} H_D \\ \Sigma_+ \\ M_1 \end{bmatrix}' = \begin{bmatrix} P(H_D,\Sigma_+,M_1) \\ Q(H_D,\Sigma_+,M_1) \\ R(H_D,\Sigma_+,M_1) \end{bmatrix}
\end{equation*}
where $P$, $Q$ and $R$ are polynomials of degree $l$, $n$ and $m$ in $H_D$ respectively, satisfying $n+1\geq l\leq m+1$. One can write this system as
\begin{equation}\begin{split}\label{E:ds3}
Q(H_D,\Sigma_+,M_1)\mathrm dH_D-P(H_D,\Sigma_+,M_1)\mathrm d\Sigma_+ &= 0 \\
R(H_D,\Sigma_+,M_1)\mathrm dH_D-P(H_D,\Sigma_+,M_1)\mathrm dM_1 &= 0,
\end{split}\end{equation}
where however the information about the direction of the flow is lost. Again the $H_D$~coordinate is formally compactified by means of a projection on the `Poincar\'e cylinder', i.e. by a coordinate transformation $(H_D,\Sigma_+,M_1)\rightarrow(X,\Sigma_+,M_1,Z):X^2+Z^2=1$ given by $H_D=\frac{X}{Z}$. Applying this to~\eqref{E:ds3}, multiplying by $Z^{m+2}$ and $Z^{n+2}$ respectively and evaluating the resulting expressions at the `equator' $X=1,Z=0$ yields
\begin{align}
\bigl(Z^mQ(X/Z,\Sigma_+,M_1)\bigr)\big|_{X=1,Z=0}\mathrm dZ &= 0\notag\\
\bigl(Z^nR(X/Z,\Sigma_+,M_1)\bigr)\big|_{X=1,Z=0}\mathrm dZ &= 0.\notag
\end{align}
Trajectories through a regular point on the `equator' corresponds to $\mathrm dZ=0$, where the flow direction is determined by the signs of $\bigl(Z^mQ(X/Z,\Sigma_+,M_1)\bigr)\big|_{X=1,Z=0}$ and $\bigl(Z^nR(X/Z,\Sigma_+,M_1)\bigr)\big|_{X=1,Z=0}$. Fixed points on the `equator' correspond to solutions of the system of equations $\{(Z^mQ)|_{X=1,Z=0}=0,(Z^nR)|_{X=1,Z=0}=0\}$. Solving this for $Q$ and $R$ in the context of~\eqref{E:dsX} yields $\Sigma_+=\pm1,M_1=0$ for the fixed points at infinity in $\mathcal X_\mathrm{VIII}$, i.e. $T_\infty:=[\infty,-1,0]^\mathrm T$ and $Q_\infty:=[\infty,1,0]^\mathrm T$.

\end{appendix}

\section*{Acknowledgements}

First and foremost, I am grateful to \mbox{J~Mark~Heinzle} who suggested this topic to me and provided guidance during the whole project. I am also thankful for the discussions I had with \mbox{S~Calogero} and \mbox{C~Uggla} during the workshop `Dynamics of General Relativity' at the Erwin Schr\"odinger Institute for Mathematical Physics in Vienna in summer~2011.


\end{document}